\documentclass{article}
\usepackage[english]{babel}
\usepackage[T1]{fontenc}
\usepackage[utf8]{inputenc}
\usepackage{amsmath,amssymb,amsthm}
\usepackage{graphicx}
\usepackage{xcolor}
\usepackage{xspace}
\usepackage{hyperref}
\usepackage{subcaption}

\ifx\pdfsuppresswarningpagegroup\undefined\else
\fi

\title{Parametric Shortest Paths in a Linearly Interpolated Graph}
\author{
Jacob Sriraman$^{1}$, Eli Barton$^{1}$, and Brittany Terese Fasy$^{1}$\\
David L.~Millman$^{2}$, Brendan Mumey$^{1}$, and Nate Rengo$^{3}$\\
Braeden Sopp$^{1}$, Vasishta Tumuluri$^{4}$, and Binhai Zhu$^{1}$\\[0.5em]
\small $^{1}$Montana State University, USA \quad $^{2}$Blocky, USA\\
\small $^{3}$Western Colorado University, USA \quad $^{4}$University of North Carolina Chapel Hill, USA
}
\date{}

\usepackage{ourmacros}
\newtheorem{theorem}{Theorem}

\begin{document}

\maketitle
\begingroup
\makeatletter
\renewcommand{\thefootnote}{\@fnsymbol\c@footnote}
\makeatother
\footnotetext[0]{This is an abstract of a presentation given at CG:YRF 2026. It has been made public for the benefit of the community and should be considered a preprint rather than a formally reviewed paper. Thus, this work is expected to appear in a conference with formal proceedings and/or in a journal.}
\endgroup

\begin{abstract}
We consider the parametric shortest paths problem in a linearly interpolated graph. Given two positively-weighted directed graphs
$G_0=(V,E,\omega_0)$ and $G_1=(V,E,\omega_1),$ the linearly interpolated graph
    is the family of graphs
    $(1-\lambda)G_0+\lambda G_1$, parameterized by $\lambda\in [0,1]$. The problem is
    to compute all distinct parametric shortest paths. We compute a data
    structure in $\Theta(k|E|\log |V|)$ time, where~$k$ is the number of
    distinct parametric shortest paths over all~$\lambda\in [0,1]$ that exist for a
    nontrivial interval of parameters, each corresponding to a linear function
    in a maximal sub-interval of $[0,1]$. Using this data structure, a shortest
    path query takes $\Theta(\log k + \ell)$ time.
\end{abstract}

\noindent\textbf{Keywords:} parametric shortest paths, linear graph interpolation, line arrangements.

\medskip
\noindent\textbf{Related version:} A full version of the paper is available at \url{https://arxiv.org/abs/2604.08892}.

\medskip
\noindent\textbf{Funding:} This research was partially supported by NSF grants CCF 2243010, CCF 2046730, and DBI 2309902.

\medskip
\noindent\textbf{Acknowledgements:} The authors would like to thank Erin Chambers, who unknowingly had a conversation that inspired the problem.

\section{Introduction}

Shortest path is a well-studied problem dating back to the 1950s, with
a wide array of applications and
algorithms such as Dijkstra~\cite{dijkstra} and
Bellman--Ford~\cite{bellman1958routing,ford1956network} covered in
undergraduate coursework. In this paper, we consider the shortest path question
in a parameterized graph, looking at weighted digraphs arising as linear
combinations of two given graphs.

We consider a geometric standpoint,
and utilize the linearity of edge weights to propose a divide-and-conquer
algorithm to compute
shortest paths between a selected source and sink in a parameterized graph, in
$\Theta(k|E|\log|V|)$ time, where~$k$ is the minimum number of shortest paths
representing all shortest paths. This, in turn, supports shortest path queries
in $\Theta(\log k + \ell)$ time, where $\ell$ is the size of the path return.

\section{Problem Statement and Preliminaries}

Consider two weighted digraphs,~$G_0=(V,E,\omega_0)$ and~$G_1=(V,E,\omega_1)$,
where $G=(V,E)$ is the
underlying digraph for both and~$\omega_0,\omega_1 \colon E \to \R_+$ are
functions.
For $\lambda \in [0,1]$, consider the interpolated graph
$G_\lambda := (V, E, \omega_\lambda)$, where~$\omega_\lambda := (1-\lambda)\omega_0 + \lambda\omega_1$.
Given two nodes~$x,y \in V$, we investigate \emph{\bf how
shortest paths from nodes $x$ to~$y$ in $G_\lambda$ evolve with~$\lambda$.} Let
$P= (e_1, e_2, \dots, e_k )$ be a path in $G$ (represented as a sequence of
edges) and~$\lambda \in \R$. Then, the cost of $P$ in graph $G_\lambda$ is
\begin{equation}
    \cost{\lambda}{P} := \sum_{e \in P} \omega_\lambda(e)
                = (1-\lambda)\sum_{e \in P} \omega_0(e)
                    + \lambda\sum_{e \in P} \omega_1(e),
\end{equation}
which we note is
a linear function of $\lambda$,
depending only on the fixed path $P$. We plot these functions and observe that
all values of shortest paths in $G_\lambda$ lie on the piecewise-linear
lower envelope of this arrangement of lines, as shown in \figref{othergraph}.

\section{Data Structures and the Algorithm}

We present a recursive solution by
finding the lower envelope of the cost functions without explicitly
computing all paths from $x$ to $y$.
The output to the algorithm is an array $M$
of interval-path pairs such that for~$([s,t],P)\in M$,~$P$ is a shortest path in
the parameterized graph for all parameter values in $ [s,t]$.

\begin{figure}
        \centering
        \begin{subfigure}[b]{0.3\textwidth}
            \includegraphics[width=.9\textwidth]{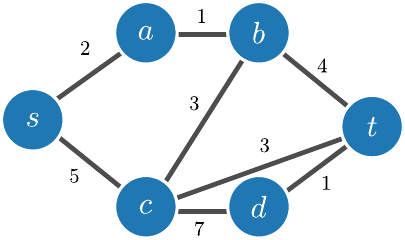}
            \caption{$H_0$}
            \label{fig:graph-0}
        \end{subfigure}
        ~
        \begin{subfigure}[b]{0.3\textwidth}
            \includegraphics[width=.9\textwidth]{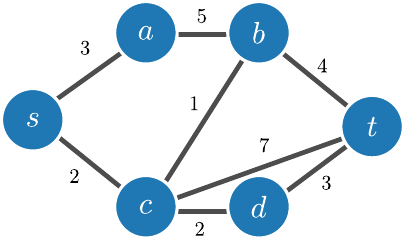}
            \caption{$H_1$}
            \label{fig:graph-1}
        \end{subfigure}
        ~
        \begin{subfigure}[b]{0.33\textwidth}
            \includegraphics[width=.9\textwidth]{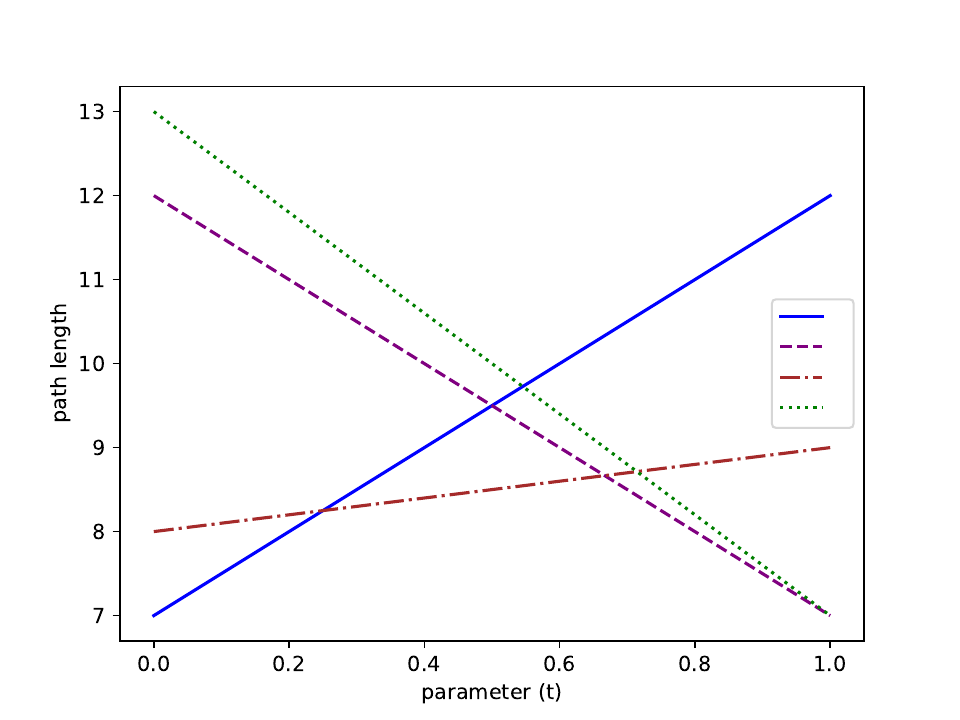}
            \caption{Path Lengths in $H_t$}
            \label{fig:graph-pathlengths}
        \end{subfigure}
        \caption{Linear interpolation on an underlying base graph $H$. We assume
        edges to be directed left to right. Our modification of Dijkstra's
    algorithm handles ties of shortest paths.}\label{fig:othergraph}
    \end{figure}

Our main algorithm is given in \algref{findle},
where
the initial call to our recursive function is
\Call{GetShortestPaths}{$G_0,G_1,x,y,0,1,\shortest{0},\shortest{1}$}, where
$\shortest{0}$ is a shortest path in~$G_0$ with minimum slope
and $\shortest{1}$ is a shortest path in $G_1$ with maximum slope.

In \lnintref{findle}{base-startif}{base-endif}, we consider the base case,
where $\shortest{s}$ is also a shortest path at $t$.
By the corresponding theorem in the full version~\cite{sriraman2026parametric},
we are done. Otherwise, we
consider the graphs of the cost functions and compute their intersection in
\lnref{findle}{intersection}; see \figref{lower-envelope}.

\[
\text{Given } \lambda, \tau \in [s,t ]\text{ we have cost function
}\cost{\lambda}{\shortest{\tau}} = \sum_{e \in \shortest{\tau}}w_s(e) -
\lambda\sum_{e \in \shortest{\tau}}w_t(e) - w_s(e).
\]

Then, we can consider the two cost functions $\cost{\lambda}{\shortest{s}} $ and
$\cost{\lambda}{\shortest{t}}$. So long as their slopes are not equal, we can
solve for their intersection point $r$.

Then, in \algtwolnref{findle}{nextpath-left}{nextpath-right}, we compute
$\shortest{r}^L$ and $\shortest{r}^R$ using a modification of
Dijkstra's algorithm selecting a single distinct shortest path that minimizes
(or maximizes) the slope of the associated cost
function depending on which interval ($[s,r]$ or $[r,t]$ respectively) we
analyze. This is done in $\Theta(|E| \log |V|)$~time. For full
details, see the full version~\cite{sriraman2026parametric}.

Finally, if $\cost{r}{\shortest{r}} = \cost{r}{\shortest{s}}$,
then $r$ is the unique breakpoint in $[s,t]$,
and $\shortest{s}$ is optimal on $[s,r]$ while
$\shortest{t}$ is optimal on $[r,t]$. If instead $\cost{r}{\shortest{r}} <
\cost{r}{\shortest{s}}$,
then we discover a new shortest path in the lower
envelope. Then, we recursively apply the procedure to
the two resulting intervals $[s,r]$ and $[r,t]$ in \lnref{findle}{recurse}.

\begin{figure}[h!]
    \centering
    \begin{subfigure}{0.30\textwidth}
        \includegraphics[width=\linewidth]{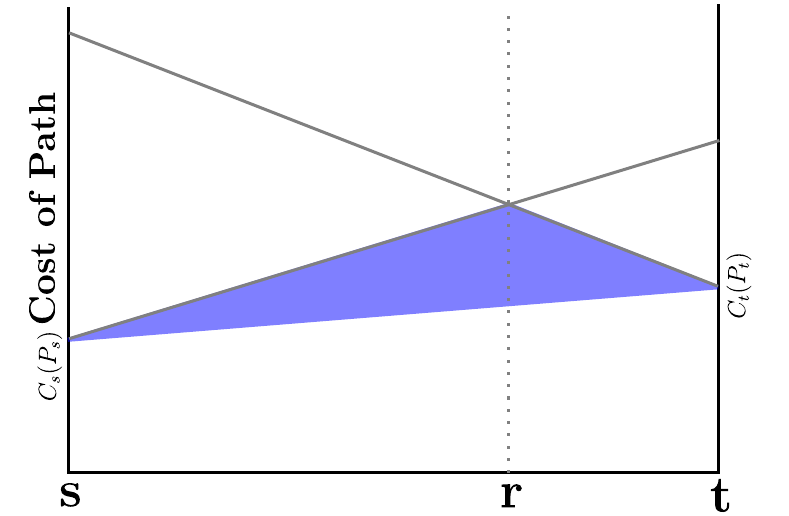}
        \caption{Before we iterate, $\cost{s}{\shortest{s}}$ and
            $\cost{t}{\shortest{t}}$ intersect
    at $r$.}
    \end{subfigure}
    \hfill
    \begin{subfigure}{0.30\textwidth}
        \includegraphics[width=\linewidth]{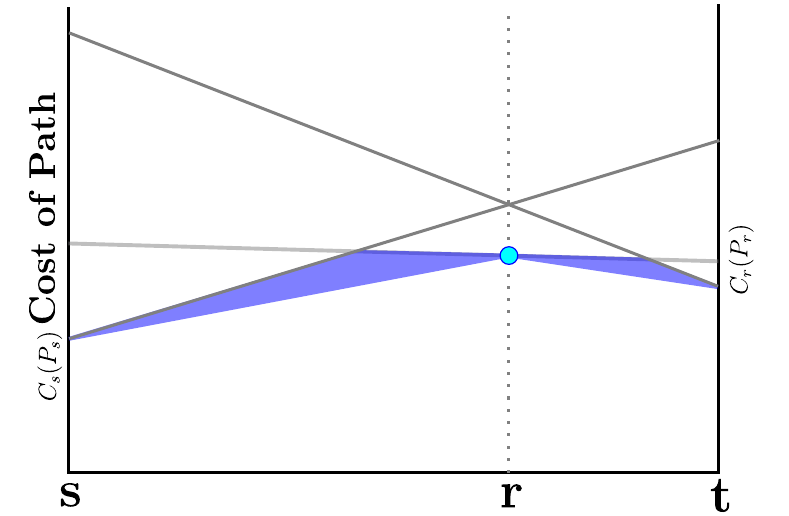}
        \caption{We compute $\cost{r}{\shortest{r}}$ and update the lower
        envelope.}
    \end{subfigure}
    \hfill
    \begin{subfigure}{0.30\textwidth}
        \includegraphics[width=\linewidth]{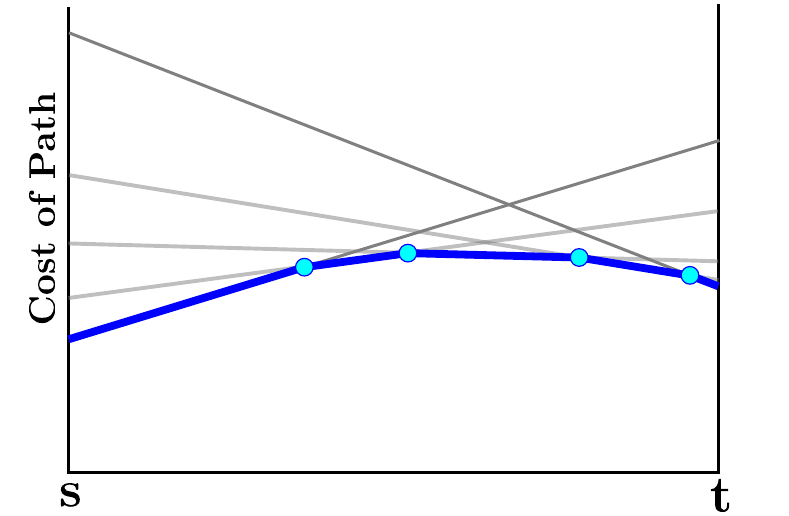}
        \caption{The algorithm terminates when lower envelope is done.}
    \end{subfigure}
    \caption{Snapshots of the algorithm before, during, and after it runs. The
        blue region
    represents the ``feasible region'' for future paths on the lower envelope.}

    \label{fig:lower-envelope}

\end{figure}

\begin{algorithm}[htb]
    \caption{\textsc{GetShortestPaths}$(G_0,G_1,x,y,s,t,\shortest{s},\shortest{t})$}\label{alg:findle}
    \begin{algorithmic}[1]
        \Require Weighted digraphs $G_0, G_1$, vertices $x,y \in V$,
        interval endpoints~$s \leq t \in \R$, a shortest path
        $\shortest{s}$ in $G_s$ of minimum slope,
        and a shortest path $\shortest{t}$ in $G_t$ of maximum~slope.
        \Ensure Sorted array of intervals with corresponding shortest paths

        \If {$\cost{t}{\shortest{s}}=\cost{t}{\shortest{t}}$}
                \label{ln:findle_base-startif}
            \State \Return $[([s,t],\shortest{s})]$
        \EndIf \label{ln:findle_base-endif}

        \State $r \gets$ value $r \in (s,t)$ such that
            $\cost{r}{\shortest{s}}=\cost{r}{\shortest{t}}$. \label{ln:findle_intersection}

        \State $P_r^L \gets \Call{DijkstraMaxSlope}{G,w_{r},\omega_1-\omega_0,x,y}$
            \label{ln:findle_nextpath-left}

        \State $P_r^R \gets \Call{DijkstraMinSlope}{G,w_{r},\omega_1-\omega_0,x,y}$
            \label{ln:findle_nextpath-right}

        \State \Return
        \Call{GetShortestPaths}{$s,r,\shortest{s},\shortest{r}$}~$\cup$~\Call{GetShortestPaths}{$r,t,\shortest{r},\shortest{t}$}
                \label{ln:findle_recurse}
    \end{algorithmic}
\end{algorithm}

Letting $T(k)$ denote the runtime when
there are $k$ lines contributing to the lower envelope.
Then,~$T(k)=T(k')+T(n-k'-1)+\Theta(|E|\log |V|)$,
which solves to~$T(k)=\Theta(k|E|\log |V|)$. The full
version~\cite{sriraman2026parametric} outlines how this
running time and algorithm compare to similar problems.

\section{Discussion}

\algref{findle} computes an ordered collection of intervals, each with an
associated linear cost function. Using this data structure, shortest-path
queries $\lambda \in (0,1)$
take~$\Theta(\log k + \ell)$ time via binary search, where $\ell$ is the size of
path returned. However, the runtime for computing
the data structure is dependent on the value of $k$.
\cite{gajjar2019lowerbound} shows that, in the worst case, the value of
$k$ is superpolynomial. Therefore, our future work includes investigating an
approximation algorithm by bounding the size of the
``feasible region,'' as defined in \figref{lower-envelope}.

\bibliographystyle{plainurl}
\bibliography{references}

\appendix

\section{Modified Dijkstra's}\label{append:mod-dij}

We present two modified versions of Dijkstra's algorithm,
$\Call{DijkstraMinSlope}{G,\omega,\delta,x,y}$
and
$\Call{DijkstraMaxSlope}{G,\omega,\delta,x,y}$,
that compute a shortest path representative in the weighted graph $(G,w)$ such
that, among all such shortest paths, the weighted path in $(G,\delta)$ is
minimized (maximized, respectively).

By our assumptions, the
edge weights in the weighted graph $(G,\omega)$ are positive.
As in Dijkstra'as algorithm,
we label each node $v$ in $V$ with a previous vertex in a path $P$ from~$s$ to
$v$; recall that the path $P$ is recovered by iteratively following `prev' until
we wind up at $s$, so it is unique for this annotated graph.
Rather than storing the length $\ell$ of $P$ alone, we store a tuple~$(\ell,m)$,
where $m$ is the slope of the cost function for $P$.
Then, each vertex has a distance-slope tuple associated to it, and
we compare these two tuples using the lexicographic ordering.
Initially, for each vertex, $\prev$ is set to null and the $\ell$ is set
to~$-\infty$. At each step, the algorithm pops $u$.
In the relaxation step, we consider an edge $(u,v)$ and an alternative path $P'$
that goes from~$s$ to~$u$ then takes the edge $(u,v)$. This path has associated
tuple~$(\ell',m')$, and we compare these two tuples using
lexicographical ordering when we wish to minimize the slope (and for maximizing
the slope, we use
$(\ell,m) \preceq (\ell',m')$ if $\ell < \ell'$ or if~$\ell
= \ell'$ and~$m' \leq m$)
Otherwise, the algorithm works just like Dijkstra's algorithm
and runs in~$\Theta(|E|\log |V|)$~time.

\section{Omitted Proof}\label{append:proofs}

\begin{theorem}[Uniform Shortest Path]\label{thm:samesp}
    Suppose that $P^*$ is a shortest path in both $G_0$ and~$G_1$. Then, $P^*$
    must be a shortest path in all $G_\lambda$ for $\lambda \in
    [0,1]$.
\end{theorem}

\begin{proof}

    Let $\pathset$ be the set of all paths from $i$ to $j$ in the base
    graph $G$, and let $P \in \pathset$ be any such path. Suppose that path $P* \in \pathset$ is a shortest path
    from $i$ to $j$ in both $G_0$ and $G_1$. Equivalently, $\forall P \in
    \pathset$
    $c_0(P^*) \leq c_0(P)$ and $c_1(P^*) \leq c_1(P).$

    Then, given $\lambda \in [0,1], (1-\lambda)c_0(P^*) \leq (1-\lambda)c_0(P)$
    and $\lambda(c_1(P^*)) \leq
    \lambda(c_1(P))$ $\forall P \in \pathset$.

    Thus, $\forall \lambda \in [0,1], \forall P \in \pathset,
    (1-\lambda)c_0(P^*) + \lambda(c_1(P^*))
    \leq (1-\lambda)c_0(P) + \lambda(c_1(P))$.

    Given $P \in \pathset$ and $\lambda \in [0,1]$, observe the following:
    \begin{align*}
    c_\lambda(P) &= \underset{e \in P}{\sum} \omega_\lambda(e) \\
           &= \underset{e \in P}{\sum} (1-\lambda)\omega_0(e) + \lambda(\omega_1(e)) \\
           &= \underset{e \in P}{\sum} (1-\lambda)\omega_0(e) + \underset{e \in
           P}{\sum} t(\omega_1(e)) \\
           &= (1-\lambda)\underset{e \in P}{\sum} \omega_0(e) +
           \lambda\underset{e \in P}{\sum} \omega_1(e) \\
           &= (1-\lambda)c_0(P) + \lambda(c_1(P))
    \end{align*}

    Thus, we have $c_\lambda(P) = (1-\lambda)c_0(P) + \lambda(c_1(P))$ for
    any $P \in \pathset$ and $\lambda \in [0,1]$. Then,
    $\forall \lambda \in [0,1]$ $\forall P \in \pathset$ $c_\lambda(P^*) \leq
    c_\lambda(P)$ i.e. the
    cost of $P^*$ in $G_\lambda$ is no more than the cost of any other path from $i$ to
    $j$ in $G_\lambda$, so it is a shortest path from $i$ to $j$.

    Thus, if $P^*$ is a shortest path from $i$ to $j$ in $G_0$ and $G_1$, then it
    must be a shortest path from $i$ to $j$ in $G_\lambda$, for any $\lambda \in [0,1]$.
\end{proof}

\section{Related Works}

Related problems have been investigated
in~\cite{DEppstein,JErickson,PAgarwal,young1991faster,Chakraborty2010twophaseparametric},
which consider graph optimization under edge weight functions that depend on
some changing parameter, with the goal of computing shortest paths more
efficiently than naively recomputing for each parameter value.
\cite{young1991faster} use Fibonacci heaps to slightly improve Karp and Orlin's
algorithm from \cite{karp1981parametric}, achieving $O(nm + n^2\log n)$ time for
a graph with $n$ vertices and $m$ edges. \cite{PAgarwal, DEppstein} study
related problems regarding spanning trees and matroid optimization.
\cite{JErickson} addresses the case of planar graphs via a parametric shortest
path tree by using a dynamic tree structure, which is similar to the methodology
in \cite{karp1981parametric}, while \cite{Chakraborty2010twophaseparametric}
employs a two-phase preprocessing approach, which allows for faster query-time
shortest path computation in exchange for more expensive~preprocessing.

Karp and Orlin study a problem most similar to ours in
\cite{karp1981parametric}. They propose an algorithm
to solve a very similar parametric shortest path problem, which corresponds to a
setting in which the edge weights in the graph $G_1$ are all~$\{0,1\}$-valued and
the parameter $\lambda$ varies in $\R$ rather than in $[0,1]$.
Their algorithm maintains an entire shortest path tree and updates it as the parameter
$\lambda$ varies. Starting from~$\lambda = -\infty$,(assuming no negative cycles)
the algorithm tracks changes in the shortest
path tree as $\lambda$ is increased, until a negative cycle is encountered,
at which point the algorithm terminates; as such, the breakpoints are processed in order.
The running time is $\Theta(|V||E| \log |V|)$.
For a sufficiently complicated lower envelope, a direct adaptation of this
algorithm
would result in a more time-efficient algorithm than what we present. However, looking
up the shortest path distance for a given pair of vertices and value of the parameter
$\lambda$ requires binary search over the breakpoints and a full traversal through
the corresponding shortest path tree. By recording only the distance
between the relevant pair of vertices rather than the full tree at each
breakpoint, this traversal cost
can eliminated and the query time becomes identical to ours.

\end{document}